\documentclass[a4paper,conference]{IEEEtran}
\IEEEoverridecommandlockouts
\usepackage{amsmath,amssymb}
\usepackage{graphicx}
\usepackage{xcolor}
\usepackage{hyperref}

\makeatletter
\newcommand{\longdash}[1][1em]{%
  \makebox[#1]{$\m@th\smash-\mkern-7mu\cleaders\hbox{$\mkern-2mu\smash-\mkern-2mu$}\hfill\mkern-7mu\smash-$}}
\makeatother
\newcommand{\omitskip}{\kern-\arraycolsep}

\usepackage[noadjust]{cite}

\ifCLASSINFOpdf
\else
\fi
\usepackage{algorithmic}

\usepackage{array}

\usepackage{multirow}
\usepackage{subfig}

\usepackage{stfloats}
\usepackage{url}

\usepackage{amsthm}

\newtheorem{theorem}{Theorem}[section]
\newtheorem{lemma}[theorem]{Lemma}

\begin{document}
%

\title{Towards Positive Jacobian: Learn to\\  Postprocess Diffeomorphic Image Registration\\with Matrix Exponential}




%
\author{\IEEEauthorblockN{Soumyadeep Pal\thanks{The first author Soumyadeep Pal is a MSc student jointly supervised by the middle author Matthew Tennant amd the last author Nilanjan Ray.}\IEEEauthorrefmark{2},
Matthew Tennant\IEEEauthorrefmark{1} and
Nilanjan Ray\IEEEauthorrefmark{2}}
\IEEEauthorblockA{\IEEEauthorrefmark{2}Department of Computing Science, \IEEEauthorrefmark{1}Department of  Ophthalmology\\
University of Alberta\\
Alberta, Canada\\
Email: \{soumyade, mtennant, nray1\}@ualberta.ca}}



\maketitle

\begin{abstract}
We present a postprocessing layer for deformable image registration to make a registration field more diffeomorphic by encouraging Jacobians of the transformation to be positive. Diffeomorphic image registration is important for medical imaging studies because of the properties like invertibility, smoothness of the transformation, and topology preservation/non-folding of the grid. Violation of these properties can lead to destruction of the neighbourhood and the connectivity of anatomical structures during image registration. Most of the recent deep learning methods do not explicitly address this folding problem and try to solve it with a smoothness regularization on the registration field. In this paper, we propose a differentiable layer, which takes any registration field as its input, computes exponential of the Jacobian matrices of the input and reconstructs a new registration field from the exponentiated Jacobian matrices using Poisson reconstruction. Our proposed Poisson reconstruction loss enforces positive Jacobians for the final registration field. Thus, our method acts as a post-processing layer without any learnable parameters of its own and can be placed at the end of any deep learning pipeline to form an end-to-end learnable framework. We show the effectiveness of our proposed method for a popular deep learning registration method Voxelmorph and evaluate it with a dataset containing 3D brain MRI scans. Our results show that our post-processing can effectively decrease the number of non-positive Jacobians by a significant amount without any noticeable deterioration of the registration accuracy, thus making the registration field more diffeomorphic. Our code is available online at \url{https://github.com/Soumyadeep-Pal/Diffeomorphic-Image-Registration-Postprocess}
\end{abstract}


%
\IEEEpeerreviewmaketitle

\section{Introduction}
 
Deformable image registration is one of the fundamental tasks of medical image analysis that has been an active research topic for decades. It constructs a dense, non-linear transformation between a pair of images to align them. A significant application of deformable registration is the alignment of 3D brain magnetic resonance (MR) images for their analysis. Brain MR images can be acquired from different sensors, different subjects or at different times and thus are misaligned. Moreover, there is often a significant variability \cite{sparks2002brain} between these scans due to different anatomical variations and health states. Deformable image registration is useful in this case for the purpose of comparing different anatomical structures in the brain scans obtained from different sources. 
 
Traditional registration algorithms are often formulated as an optimization problem where a moving image is warped using a displacement field and the goal is to maximize the similarity between a fixed image and the warped moving image. This is usually solved using an iterative process, which is fairly computationally intensive and time consuming. However, recently deep learning approaches have been used in solving the deformable registration problem. The deep learning approaches maintain similar performance in terms of registration accuracy and are much faster.



One of the desirable properties of the transformations for registration in medical imaging is their one-to-one nature or invertibility, which ensures that there is no folding in the grid. Folding of the registration grid over itself can lead to connected sets becoming disconnected and disconnected sets becoming connected thereby destroying the neighbourhood structure that is detrimental for anatomical studies in medical imaging \cite{beg2005computing}. Classical diffeomorphic registration algorithms often use strategies, which ensure smooth, invertible transformations. However, deep learning based registration methods usually do not explicitly ensure invertibility and non-folding of the transformations. Such foldings in the deformations are usually constrained by enforcing spatial smoothness, which is controlled by a regularization hyperparameter. However, a large value of this hyperparameter can lead to inaccurate registration, while a small value can lead to folding and local errors, which makes it challenging to tune it.

 In this paper, we explicitly address the issue of folding with a postprocessing layer, which can be potentially inserted at the end of any registration pipeline giving a deformation field as its output. Our postprocessing step takes a deformation field as its input and provides another deformation with reduced foldings as its output with the help of matrix exponential and Poisson reconstruction. Moreover, this postprocessing layer is completely differentiable, hence it can fit in any deep learning pipeline for registration with end-to-end learning. We demonstrate the effectiveness of our method by the registration of 3D brain MR scans, which are obtained from \cite{Marcus2007}. We use our postprocessing layer with a widely used registration method named Voxelmorph \cite{balakrishnan2018unsupervised} and show that it significantly reduces the amount of folding when compared to Voxelmorph.

\section{Background}

In registration, we typically have a moving image, which we want to align with a fixed/reference image. Deformable image registration entails warping the moving image with a dense voxel-wise non-linear spatial transform so that it matches with a fixed image in terms of a similarity metric. This is in contrast to the rigid/affine registration paradigm where we have a linear rigid-body transformation (like rotation and translation). A typical registration process involves an affine transformation for a global alignment between the fixed and moving image and then a deformable registration step. Let $F$ be the fixed image and $M$ be the moving image. Using a displacement field $\phi$ to warp $M$, the objective of deformable registration is to find an optimal displacement field:
\begin{equation}
    \phi^* = \mathop{\arg \min}\limits_{\phi} \mathcal{L}_{sim}(F,M(\phi)) + \lambda \mathcal{L}_{reg}(\phi)
    \label{regis}
\end{equation}
where $\phi^*$ is the optimal displacement field, $\mathcal{L}_{sim}$ is a dissimilarity loss function and $\mathcal{L}_{reg}$ is the function that enforces a smoothness regularization. There are different dissimilarity functions that are typically used like mean square voxel difference, negative cross-correlation and mutual information to minimize the dissimilarity between the fixed image and moving image warped with $\phi$. The regularization term is often a norm of the gradient of the displacement field, but it can also be used for curvature regularization. 


However, optimizing (\ref{regis}) does not ensure that the registration field is diffeomorphic and thus it does not have the desirable properties of invertibility/non-folding. Diffeomorphic frameworks usually achieve that by considering a registration field as the integral of a velocity vector field \cite{ashburner2007fast}. In such cases, the deformation field is an exponential map of the vector field and lies in a Lie group, thus making the deformation diffeomorphic \cite{ashburner2007fast}.

In this paper, we focus on reducing the number of non-positive Jacobians of the final registration field. Jacobian matrices of the deformation field encode local volume changes and a negative determinant of the Jacobian matrices indicates loss of invertibility, which results in folding of the registration grid \cite{ashburner2007fast}. Positive Jacobians preserve the orientation of a transformation map and guarantee invertibility. One potential way to tackle negative Jacobians is to simply add a penalty in the optimization objective / loss function similar to \cite{kervadec2019constrainedinequality}. However that is a soft constraint for the inequality of the Jacobian being positive. Lagrangian dual optimization \cite{boyd2004convex} is a standard method to handle inequalities in an optimisation problem. However, it is generally avoided in deep learning due to stability issues and computational burden \cite{marquez2017imposing}. Log barrier extensions \cite{kervadec2019constrained} may be used to approximate the Lagrangian optimization, however it may be a non-trivial way to address the issue. Here, we introduce a simpler postprocessing layer that can successfully increase the number of positive Jacobians and can be easily plugged into a deep learning framework.

\section{Related Works}

Classical registration algorithms often model deformations as a physical model like elasticity and fluid flow. They perform iterative optimization to minimize an energy functional that is similar to (\ref{regis}) and is based on the physical model. These include the elastic type models (\cite{shen2002hammer,gefen2003elastic, bajcsy1989multiresolution}), fluid flow models (\cite{christensen1996deformable,bro1996fast}). Different registration methods change the space of optimization by parameterizing the displacement field. Such algorithms include free form deformation with b-splines \cite{rueckert1999nonrigid}, radial basis spline methods \cite{fornefett2001radial}, thin plate splines (\cite{vserifovic2009intensity,zhen2015segmentation}). Other methods for deformable registration include Demons \cite{thirion1998image}, statistical parameter mapping \cite{hellier2002inter}, DRAMMS \cite{ou2011dramms}. Some registration algorithms perform the optimization by constraining it in the space of diffeomorphic maps, thus giving a desirable diffeomorphic registration. Popular methods under this category include LDDMM \cite{beg2005computing}, diffeomorphic demons \cite{vercauteren2009diffeomorphic}, DARTEL \cite{ashburner2007fast}, standard symmetric normalization (SyN) \cite{avants2008symmetric}.

Due to the advent of deep learning, recent works have used convolutional neural networks to perform registration. These methods can be broadly divided into supervised and unsupervised ones. Supervised deep learning methods for deformable registration usually learn a deep learning model estimating displacement fields using ground truth images (\cite{krebs2017robust,yang2017quicksilver,sokooti2017nonrigid,rohe2017svf,cao2017deformable}). However, this may result in transformations biased by ground truths and in practical problems, it is very difficult to obtain a large amount of ground-truth information. Thus, unsupervised deep learning approaches have also been developed, which will be the focus of this paper.

\begin{figure*}[!h]
\centering
\includegraphics[width=7in]{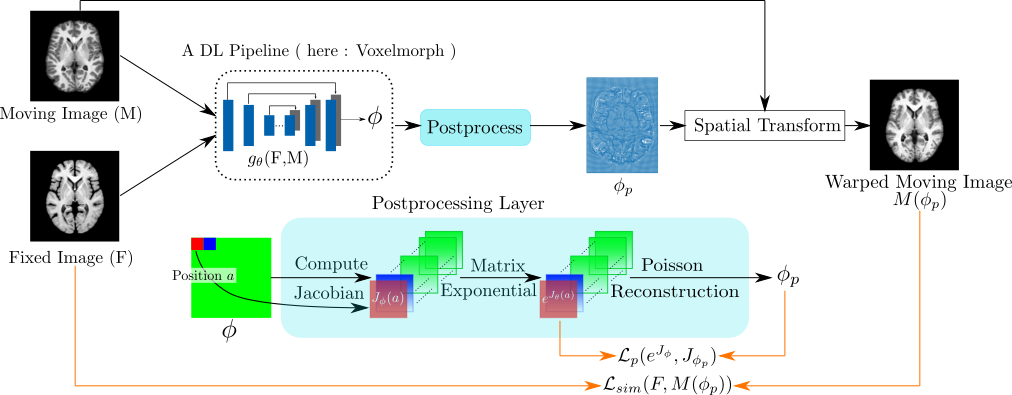}
\label{JacFilt}
\caption{Overview of the proposed postprocessing layer for Diffeomorphic Image Registration. The output from the registration pipeline $\phi$ is the input to out proposed layer which gives $\phi_p$ as its output. $J_\phi$ indicate the Jacobian matrices of a displacement field $\phi$. We omit the regularization loss for simplicity}
\end{figure*}

Typically, unsupervised deep learning methods approximate the displacement field with a CNN, warp the moving image using a spatial transformer and minimize a similarity metric similar to (\ref{regis}). Vos et.al.\cite{de2017end} consider corresponding patches from the fixed and moving image as inputs to the CNN, a b-spline transformer and train with the cross-correlation similarity metric. Li et.al.\cite{li2017non} also consider a fully convolutional network and also uses the cross-correlation similarity metric along with a total variation based regularizer. These methods demonstrate their efficiency using 2D images or small 3D image regions. Balakrishnan et.al.\cite{balakrishnan2018unsupervised} present a CNN model, popularly called Voxelmorph, a spatial transformation function based on spatial transformer networks \cite{jaderberg2015spatial} and optimize a cross-correlation based similarity metric along with a gradient regularizer. 

These deep learning methods do not typically consider the diffeomorphic properties of the registration field. There have been few works that look into learning diffeomorphic registration fields. Dalca et.al. \cite{dalca2018unsupervised} use a generative CNN model to estimate the distribution for a velocity field and find a diffeomorphic deformation field by integration of the velocity field by scaling and squaring. Mok et.al. \cite{mok2020fast} use a similar scaling and squaring theme, but also develop a symmetric similarity framework considering both the forward and backward transformation between a pair of images. Kuang et.al. \cite{kuang2019faim} use a penalty loss to constrain non-positive Jacobians and train their CNN with a use cross correlation similarity loss. \cite{dalca2018unsupervised} and \cite{mok2020fast} mainly use scaling and squaring based velocity field integration to produce diffeomorphic registration fields. In this paper, we focus on developing a postprocessing layer that can reduce the number of non-positive Jacobians and can potentially fit in any deep learning registration pipeline.

\section{Method}
Learning algorithms for diffeomorphic registration usually do not look into explicitly preventing folding of the transformation in different voxel locations. In our setting, we tackle this issue with the help of matrix exponential.  

Let $F, M$ be two 3D image volumes. As mentioned before, the two image volumes are initially aligned with a global deformation as a preprocessing step such that the remaining misalignment between $F$ and $M$ is non-linear.  

A function $g_{\theta}(F,M)$ is modelled by a convolutional neural network (CNN), such that the neural network outputs a displacement field or registration field $\phi$ \cite{balakrishnan2018unsupervised}. The displacement field is a four dimensional vector that determines the displacement between $F$ and $M$. Considering $Id$ as an identity transform, the transformaton $Id + \phi$ is used to warp the moving image, such that for each voxel location $q$, $F(q)$ and $M(\phi(q))$ are identical.  

For our experiments we use the Voxelmorph-2 architecture described in \cite{balakrishnan2018unsupervised} as our neural network. We insert a postprocessing layer at the end of the neural network, which takes in the displacement field $\phi$ and gives a new displacement field $\phi_p$ as an output, which potentially contains much less folding when compared to $\phi$.

\begin{figure*}[!h]
\centering
\includegraphics[width=6in]{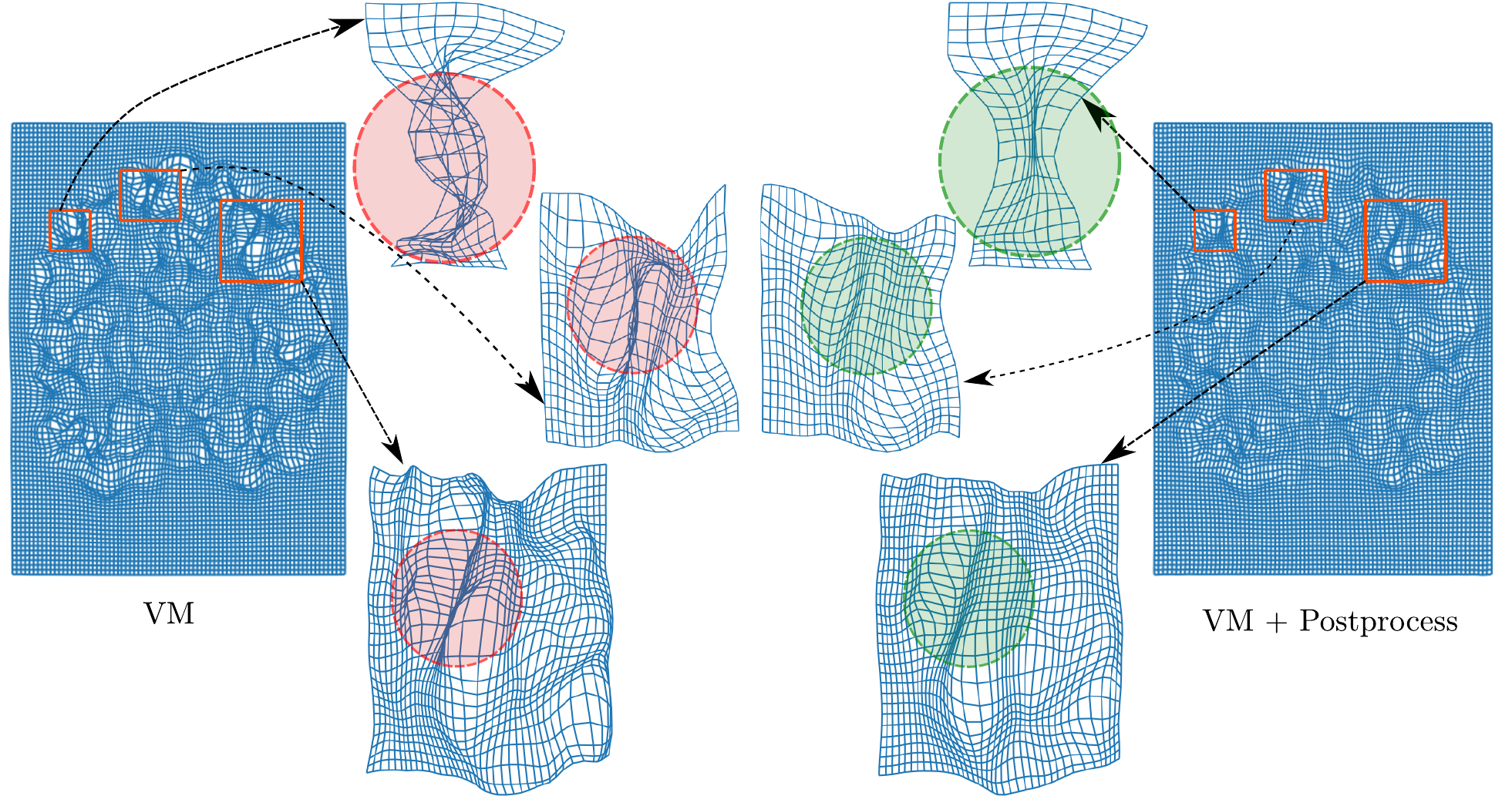}
\label{JacFilt}
\caption{Illustration showing the effect of the Postprocessing Layer on foldings in the grid used for registration for one slice. VM = Voxelmorph. Red circles indicate folding of the grid over itself. Green circles indicate no foldings in same region.}
\end{figure*}

\subsection{Postprocess Formulation}
As shown in Fig.1, the postprocessing step essentially finds the Jacobian matrix of the displacement field, computes its exponential and reconstructs a new displacement field from the exponentiated Jacobian matrix.

Let $\phi , \phi_p \in \mathbb{R}^{H \times D \times W \times 3}$ be displacement fields and $q = (x,y,z)$ be a voxel location. $H, D, W$ respectively, refer to height, width and depth of the image volumes. The Jacobian matrix of the displacement field $\phi = (\phi_x, \phi_y, \phi_z)$ at $q$ is defined as:
\begin{equation}
    Jac(\phi(q)) = \begin{bmatrix}
     \nabla\phi_x(q)\\
     \nabla\phi_y(q)\\
     \nabla\phi_z(q)
    \end{bmatrix}  
    = \begin{bmatrix}
     \frac{\partial \phi_x(q)}{\partial x} & \frac{\partial \phi_x(q)}{\partial y} & \frac{\partial \phi_x(q)}{\partial z} \\
     \frac{\partial \phi_y(q)}{\partial x} & \frac{\partial \phi_y(q)}{\partial y} & \frac{\partial \phi_y(q)}{\partial z} \\
     \frac{\partial \phi_z(q)}{\partial x} & \frac{\partial \phi_z(q)}{\partial y} & \frac{\partial \phi_z(q)}{\partial z}
    \end{bmatrix}  
\end{equation}
Using matrix exponential, we get the matrix $J'$ :
\begin{equation}
    J'(q) = \begin{bmatrix}
    J'_x(q) \\
    J'_y(q) \\
    J'_z(q)
    \end{bmatrix}   = e^{Jac(\phi(q))}
    \label{Jexp}
\end{equation}
We compute the matrix exponential of the $3 \times 3$ matrix $Jac(\phi(q))$ by a series summation scheme as mentioned in \cite{nan2020drmime}.
The postprocessed displacement field is reconstructed by solving the following Poisson's equations with Dirichlet boundary conditions:
\begin{equation}
 \begin{split}
    &\Delta \phi_{p_x} = \nabla \cdot  J'_x \\
    &\Delta \phi_{p_y} = \nabla \cdot  J'_y \\
    &\Delta \phi_{p_z} = \nabla \cdot  J'_z \\
 \end{split}
 \label{Pois}
\end{equation}
The value of the displacement field at the boundaries is considered 0, because of the nature of the images.

We solve each Poisson's equation adapting from \cite{gilbertcse} using Discrete Sine Transform (DST).
Computing the right hand side of (\ref{Pois}) is straightforward because we have $J'_x, J'_y, J'_z$ from (\ref{Jexp}).

The solution of the Poisson's equations is given by
\begin{equation}
 \begin{split}
    \phi_{p_t}(i,j,k) = \sum_{l=1}^H \sum_{m=1}^D & \sum_{n=1}^W \frac{a_{lmn}}{\lambda_{lmn}} \sin{\frac{il\pi}{H+1}}\\
    &\sin{\frac{jm\pi}{D+1}}\sin{\frac{kn\pi}{W+1}}, \\
    \text{for} \ t=x,y,z,
 \end{split}
\end{equation}
where $\lambda_{lmn}$ are the eigenvalues of a Laplacian matrix as follows:
\begin{equation}
 \begin{split}
    \lambda_{lmn} = ( 2 - 2 \cos{\frac{l\pi}{H+1}}) + (2 - 2 \cos{\frac{m\pi}{D+1}}) + \\
    (2 - 2 \cos{\frac{n\pi}{W+1}})
 \end{split}
\end{equation}

The values of $a_{lmn}$ can be found from the following relation based on DST:
\begin{equation}
 \begin{split}
   \nabla \cdot  J'_{t} = \sum_{l=1}^H \sum_{m=1}^D & \sum_{n=1}^W a_{lmn}\sin{\frac{il\pi}{H+1}}\\
    &\sin{\frac{jm\pi}{D+1}}\sin{\frac{kn\pi}{W+1}}.
 \end{split}
\end{equation}
Thus the solution of (\ref{Pois}) can be found by computing the DST of the divergence of $J'$, dividing the result by the corresponding eigenvalues and computing the inverse DST of the result. This makes the reconstruction of $\phi_p$ computationally efficient. Details of this numerical method appear in \cite{gilbertcse}. The solution is efficient because DST is a separable operation.

\subsection{Postprocess Analysis}

In this postprocessing layer, we exploit the properties of matrix exponential and Poisson reconstruction to achieve a reduced number of non-positive Jacobians. The matrices $J'$ in \ref{Jexp} all have positive determinants, because of the following property of matrix exponential:
\begin{equation}
 \begin{split}
   det(J'(q)) = det(e^{Jac(\phi(q))}) = e^{tr(Jac(\phi(q)))} > 0 
 \end{split}
\end{equation}
If these matrices $J'$ are valid Jacobian matrices, i.e. if they are integrable, then a perfect reconstruction of a field from these will give us a displacement field with all positive Jacobians. In Appendix, we show that under certain conditions, we can get such integrable matrices.

The unique solutions of the Poisson equations in (\ref{Pois}) is the solutions of the minimization problems: \cite{perez2003poisson}:
\begin{equation}
 \begin{split}
   &\mathop{\min}\limits_{\phi_{p_x}} ||\nabla \phi_{p_x} - J'_x||^2_2, \\  
   &\mathop{\min}\limits_{\phi_{p_y}} ||\nabla \phi_{p_y} - J'_y||^2_2, \\  
  &\mathop{\min}\limits_{\phi_{p_z}} ||\nabla \phi_{p_z} - J'_z||^2_2. \\  
 \end{split}
 \label{PoisLoss}
\end{equation}
Thus, the Poisson reconstruction ensures that the Jacobian matrices of the reconstructed field $\phi_p$ are close to matrices with positive determinants, even if $J'_{t}$s are not integrable.

\begin{figure*}[!h]
\centering
\includegraphics[width=7in]{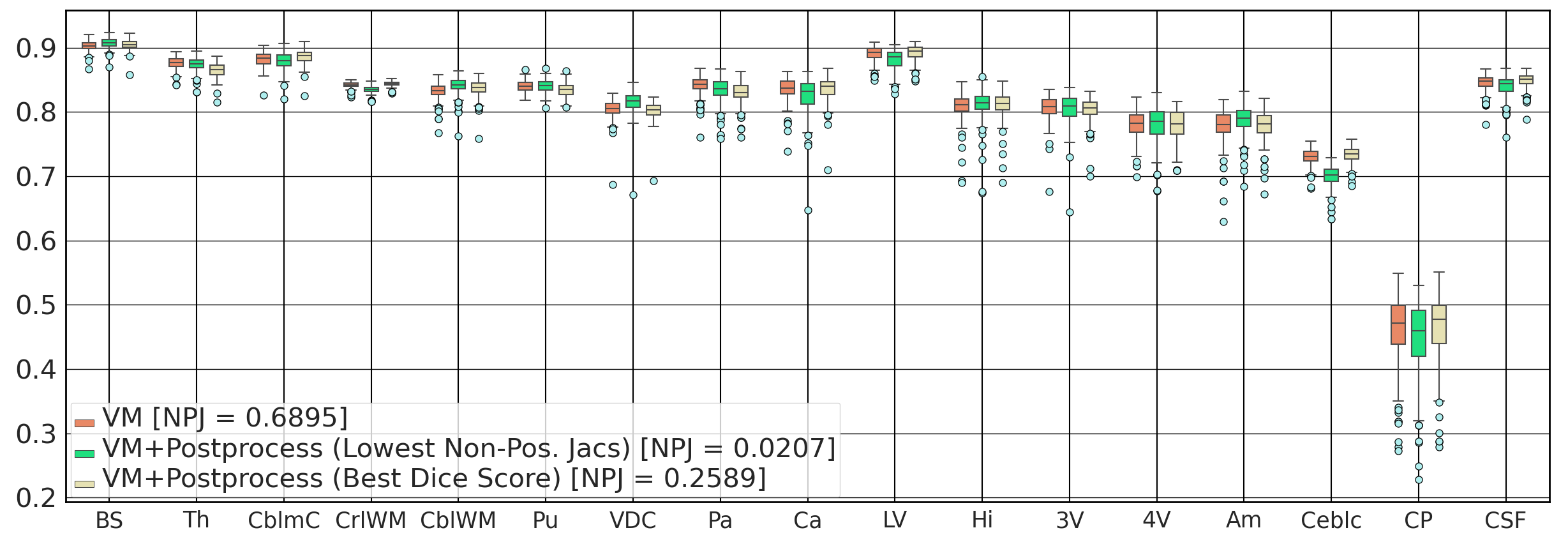}
\label{JacFilt}
\caption{Box Plot of Dice Scores of different anatomical structures for VM, model of VM with Postprocessing Layer giving best Dice Score and VM with Postprocessing Layer giving lowest percentage of non-positive Jacobians. NPJ = Percentage of non-positive Jacobians. Structures with left and right hemispheres are combined into one for this illustration. Anatomical structures: Brain Stem (BS), Thalamus (Th), Cerebellum Cortex (CblmC), Cerebral White Matter (CrlWM), Cerebellum WM (CblWM), Putamen (Pu) , Ventral-DC (VDC), Pallidum (Pa), Caudate (Ca), Lateral Ventricle (LV), Hippocampus (Hi), 3rd Ventricle 3V), 4th Ventricle (4V), Amygdala (Am), Cerebral Cortex (Ceblc), Choroid Plexus (CP) and CSF.}
\end{figure*}

\subsection{Poisson Reconstruction Loss}

The Poisson reconstruction step plays a crucial role in reducing the non-positive Jacobians. Note that a lower loss value in (\ref{PoisLoss}) implies a better integrability condition for the matrix $J'$. Hence, adding the three losses in (\ref{PoisLoss}), we introduce a Poisson reconstruction loss $L_p$ as follows:
\begin{equation}
L_p = \sum_{q} \|e^{Jac(\phi(q))} - Jac(\phi_p(q))\|_2^2
\end{equation}

In our experiments, we use a loss function similar to the Voxelmorph loss \cite{balakrishnan2018unsupervised} along with our Poisson reconstruction loss. Following \cite{balakrishnan2018unsupervised}, the similarity measure $L_{sim}$ is given by the negative cross correlation (CC) and the regularization loss $L_{reg}$ is a diffusion regularizer over the spatial gradients. Hence, the complete loss is as follows:
\begin{equation}
\begin{split}
    L(F,M) = -CC(F,M(\phi_p)) +  \lambda \sum_{p}\|\nabla\phi_p(p)\|^2
    + \lambda_p L_p,
\end{split}
\end{equation}
where $\lambda_p$ is the hyperparameter determining the strength of the Poisson reconstruction loss.



\section{Experiments and Results}
\subsection{Data}
In this paper, we use the open-access OASIS dataset \cite{Marcus2007} to evaluate our postprocess step. The dataset contains 414 T1-weighted brain MRI scans from subjects aged 18 to 96. We obtain the preprocessed dataset from \cite{adalca}. The MRI scans were preprocessed \cite{hoopes2021hypermorph} using Freesurfer \cite{fischl2012freesurfer} by standard steps like resampling, bias correction, skull stripping, affine normalization and center cropping into volumes of $160 \times 192 \times 224$. For our experiments we split the dataset into training, validation and test set of sizes 255, 15 and 144 respectively.

We perform atlas-based registration for our experiments i.e. we aim to establish anatomical correspondence between the moving images and the reference image/atlas. An atlas can be a single volume or an average of volumes in the same image space. Atlas-based registration is commonly applied to register inter-subject images. In this paper, we use an atlas constructed from a different dataset \cite{sridharan2013quantification} and also used in the official implementation of Voxelmorph \cite{balakrishnan2018unsupervised}.

\subsection{Evaluation Metric}

We evaluate the performance of our postprocessing layer with two metrics: Dice Score (DS) and the percentage of non-positive Jacobian determinants ($|J_{\phi_p}|\leq 0$)  \cite{mok2020fast}. The Dice Score measures the volume overlap of different segmented anatomical structures. Considering $S^k_F$ and $S^k_{M(\phi_p)}$ to be the sets of voxels for an anatomical structure $k$ for $F$ and $M(\phi_p),$ respectively, the Dice Score \cite{dice1945measures} is given by :
\begin{equation}
    DS(S^k_F, S^k_{M(\phi_p)}) = 2 \times \frac{S^k_F \cap S^k_{M(\phi_p)}}{|S^k_F|+|S^k_{M(\phi_p)}|}
\end{equation}
For our analysis, we consider 30 anatomical structures for the computation of Dice Score \cite{balakrishnan2019voxelmorph}. A Dice Score of 1 is highest since it indicates complete overlap and no overlap gives the lowest score of 0. The goal of our postprocessing step is to reduce the number of non-positive Jacobians. Hence we also measure the percentage of voxels which have non-positive Jacobian determinants to evaluate our layer.

\subsection{Implementation}

Since we propose a postprocessing layer in this paper, we compare the results between an existing framework, namely the voxelmorph framework (VM) \cite{balakrishnan2018unsupervised} and that of our layer used in conjunction with the voxelmorph framework. We implement our method using Pytorch \cite{paszke2017automatic} . We train our models using the Adam optimizer \cite{kingma2014adam} with a learning rate of $1e^{-4}$. We train both VM and VM in conjunction with our proposed postprocess layer and tune the respective hyperparameters with grid search. Based on the best dice score from the validation set, we get the best result for our setting using $\lambda = 1.0$ and $\lambda_p = 0.01$.

\begin{table}[!t]
\renewcommand{\arraystretch}{1.3}
\caption{Average Dice Score (higher is better) and Average Percentage of Non-Positive Jacobians (lower is better) with $\lambda = 1, 2$ and increasing $\lambda_p$. VM = Voxelmorph. Standard Deviation Given in Parenthesis}
\label{table_example}
\centering
\begin{tabular}{c  c  c  c c}
\hline
\bfseries Method & \bfseries $\lambda$ & \bfseries $\lambda_{p}$ & \bfseries Avg. Dice & \bfseries \% of $|J_{\phi_p}|\leq 0$\\
\hline\hline
\      VM           &            1.0       & \multicolumn{1}{l}{-} & \multicolumn{1}{l}{0.8056 (0.0084)} & \multicolumn{1}{l}{0.6895 (0.0950) } \\\cline{1-5} 
                      & \multirow{4}{*}{$1.0$}  & \multicolumn{1}{l}{0} & \multicolumn{1}{l}{0.8051 (0.0081)} & \multicolumn{1}{l}{0.2964 (0.0461)} \\\cline{3-5}
         VM +          &           & \multicolumn{1}{l}{0.01} & \multicolumn{1}{l}{0.8053 (0.0082) } & \multicolumn{1}{l}{0.2589 (0.0411)} \\\cline{3-5}
     Postprocess                  &           & \multicolumn{1}{l}{0.05} & \multicolumn{1}{l}{0.8045 (0.0087)} & \multicolumn{1}{l}{0.2255 (0.0378)} \\\cline{3-5}
                       &                       & \multicolumn{1}{l}{0.1} & \multicolumn{1}{l}{0.8024 (0.0095)} & \multicolumn{1}{l}{0.1154 (0.0280)} \\
\hline
VM                   &      2.0        & \multicolumn{1}{l}{-} & \multicolumn{1}{l}{0.8048 (0.0086)} & \multicolumn{1}{l}{0.2881 (0.0501)} \\\cline{1-5} 
                    &   \multirow{4}{*}{$2.0$}    & \multicolumn{1}{l}{0} & \multicolumn{1}{l}{0.8025 (0.0092)} & \multicolumn{1}{l}{0.0639 (0.0187)} \\\cline{3-5}
       VM +        &               & \multicolumn{1}{l}{0.01} & \multicolumn{1}{l}{0.8032 (0.0100)} & \multicolumn{1}{l}{0.0609 (0.0190)} \\\cline{3-5}
      Postprocess       &           & \multicolumn{1}{l}{0.05} & \multicolumn{1}{l}{0.8043 (0.0096)} & \multicolumn{1}{l}{0.0672 (0.0201)} \\\cline{3-5}
                    &                    & \multicolumn{1}{l}{0.1} & \multicolumn{1}{l}{0.8025 (0.0103)} & \multicolumn{1}{l}{0.0207 (0.0001)} \\
\hline
\end{tabular}
\end{table}

\begin{table}[!t]
\renewcommand{\arraystretch}{1.3}
\caption{Average Run Time in secs for Registration of Pair of Images (lower is better). VM = Voxelmorph. Standard Deviation Given in Parenthesis}
\label{table_example}
\centering
\begin{tabular}{c c}
\hline
\bfseries Method & \bfseries Time in sec\\
\hline\hline
VM  & 0.60 (0.10)\\
VM+Postprocess & 1.87 (0.20)\\
\hline
\end{tabular}
\end{table}

\subsection{Registration Performance}

Table I shows the average dice score and the average percentage of voxels with non positive Jacobians for all subjects in the test set for our experiments for different values of $\lambda$ and $\lambda_p$. We observe that adding our postprocessing layer does not noticeably alter the dice score performance, however it reduces the percentage of non-positive Jacobians by a significant amount. Thus our proposed postprocessing layer can reduce folding (see Fig. 2 for an example) of the registration grid while maintaining a high registration accuracy (in terms of dice score), thus giving more diffeomorphic transformations. 

In Fig. 3, we also show the average dice score for different anatomical structures in the brain as a boxplot. We demonstrate that for VM, VM with postprocessing layer giving the best Dice and VM with postprocessing layer giving the lowest percentage of non-positive jacobians.

\subsection{Effect of the Poisson Reconstruction Loss}

We also demonstrate the effect the proposed reconstruction loss in Table I. We show the average percentage of non-positive jacobians for increasing values of $\lambda_p$ with the gradient regularization $\lambda = 1,2$. As the weight of the reconstruction term increases through $\lambda_p$, we observe that amount of non-positive Jacobians decreases; however there is not much of a decrease in Dice score. Thus, with increasing $\lambda_p$, the reconstruction loss tries to reconstruct a displacement field $\phi_p$, whose Jacobian is increasingly closer to $e^{J_\phi}$ and thus is more diffeomorphic. Hence, our proposed loss is successful in making the deformations more diffeomorphic without sacrificing too much registration accuracy. 

\subsection{Runtime Analysis}

Table II shows the average time required to register a pair of images when we use a VM trained model and when we use a VM model trained along with our layer. We perform the deformable registration of a MRI scan of a test subject to the atlas using a NVIDIA Tesla P100 GPU and an Intel Xeon (E5-2683 v4) CPU. The runtime for registration when we add our layer is greater than that for VM by just about 1.2 seconds. Thus, it maintains the advantage of deep learning methods being faster than traditional registration methods.

\section{Conclusion}
In this paper we have presented a postprocessing layer which can fit in a deep learning registration framework with end-to-end learning. We evaluate our layer using large scale brain MR dataset with the Voxelmorph framework and show that our layer is successful in reducing folding in the registration grid and maintaining high registration accuracy. Even though we employ a Poisson equation solver, our layer still maintains the advantage of fast registration, a desirable characteristic of deep learning algorithms. We hope that our postprocessing layer can be used in other registration frameworks desiring more diffeomorphic registration fields.

The exponentiated Jacobian is not always integrable, but under certain conditions it can give a valid Jacobian as explored in the Appendix and that will lead to a theoretical guarantee of strictly positive Jacobians. Thus, for future work, we hope to develop constrained registration fields that can lead to a theoretically guaranteed, fully diffeomorphic registration.

\appendix

\begin{lemma} If $J$ is a Jacobian matrix of a conservative vector field, then
$e^J$ will also be a Jacobian matrix of a conservative vector field, when there is a $3 \times 3$ matrix $A$, such that $\frac{\partial J}{\partial x}J=A \frac{\partial J}{\partial x}$, $\frac{\partial J}{\partial y}J=A \frac{\partial J}{\partial y}$, and $\frac{\partial J}{\partial z}J=A \frac{\partial J}{\partial z}$.
\end{lemma}

\begin{proof}
For brevity of space we provide an outline of the proof here. Using the differentiation formula \cite[p.115]{Brian2015}, we get:
\begin{equation}
\begin{split}
    \frac{\partial}{\partial t}e^J = ~& e^J\{\frac{I-e^{-ad_J}}{ad_J}\}\frac{\partial J}{\partial t} \\
    = ~& e^J(\frac{\partial J}{\partial t}-\frac{1}{2}[J,\frac{\partial J}{\partial t}]+\frac{1}{6}[J,[J,\frac{\partial J}{\partial t}]]-... ) \\
    = ~& e^JB\frac{\partial J}{\partial t},~t=x,y,z,
\end{split}
\end{equation}
for a $3 \times 3$ matrix $B.$ Here, $[.,.]$ denotes the Lie bracket. Now, we can easily verify that the matrix $e^J$ is curl-free.
\end{proof}
Even though the set of conditions in the lemma is technical, it is quite general and we point out that such conditions hold for a large family of functions. For example, for the 2D cases, we can verify that any harmonic function and its conjugate \cite{Gamelin2001} together obey these conditions and provide us with such conservative vectors fields. Non-trivial families of functions following these conditions also exist in 3D.





\bibliographystyle{IEEEtran}
\bibliography{IEEEexample}
%

\end{document}